\documentclass[10pt, conference]{IEEEtran}
\IEEEoverridecommandlockouts
\usepackage{ifpdf} %%切换latex和pdflatex命令编译
\usepackage{array} %%%排列包（创建排列，1、2、3这种排列）

\usepackage{amsmath} %%数学包
\usepackage{amsthm} % 如果要使用proof语句就必须要引入这个包

\theoremstyle{remark}
\newtheorem{theorem}{Theorem}

\newtheorem{lemma}{Lemma}

\newtheorem{proposition}{Proposition}
\newtheorem{remark}{Remark}

\usepackage{amssymb}    %
\usepackage{amsfonts}
\usepackage[english]{babel}
\usepackage{cite}
\usepackage{multirow}

\usepackage{makecell}
\usepackage{stfloats}    % big equations float
\usepackage{algorithm}   %format of the algorithm
\usepackage{algorithmic} %format of the algorithm
\usepackage{graphicx}
\usepackage[nooneline]{caption}
\usepackage{epsfig}
\usepackage{subfigure} %subfigs
\usepackage{cite}

\usepackage{setspace} %for begin{spacing} for complex tables
\usepackage{threeparttable}

\usepackage{url}

\usepackage{footnote}
\usepackage[textsize=tiny]{todonotes}
\usepackage{marginnote}
\usepackage{pifont}
\usepackage[perpage]{footmisc}  %每页脚注重新编号

\usepackage{booktabs}    %表格

\allowdisplaybreaks  %allow equation crossing pages

 \usepackage[draft]{changes}   %审阅模式 % draft or final

\usepackage{color}

\begin{document}
\title{Age of Information of CSMA/CA Based \\Wireless Networks}
\author{\IEEEauthorblockN{Liang Li\textsuperscript{1}, {Yunquan Dong}\textsuperscript{1,2},~\IEEEmembership{Member,~IEEE}, Chengsheng Pan\textsuperscript{1}
and {Pingyi Fan}\textsuperscript{3},~\IEEEmembership{Senior Member,~IEEE}}

\IEEEauthorblockA{\textsuperscript{1}School of Electronic \& Information Engineering, Nanjing University of Information Science \& Technology, Nanjing, China}
\IEEEauthorblockA{\textsuperscript{2}National Mobile Communications Research Laboratory, Southeast University, Nanjing, China}
\IEEEauthorblockA{\textsuperscript{3}Beijing National Research Center for Information Science and Technology (BNRist), Beijing, China \\Department of Electronic Engineering, Tsinghua University, Beijing, China}
\IEEEauthorblockA{\{liangli, yunquandong, 003150\}@nuist.edu.cn, fpy@tsinghua.edu.cn}

\thanks{The work of Y. Dong was supported by the National Natural Science Foundation of China (NSFC) under Grant 62071237 and 61931004, the open research fund of National Mobile Communications Research Laboratory, Southeast University, under grant No.2020D09.
        The work of P. Fan was supported by the National Key R\&D Program of  China under grant No.2021YFA1000504.
        }
        }

\maketitle

\begin{abstract}
 We consider a wireless network where $N$ nodes compete for a shared channel over the CSMA/CA protocol to deliver observed updates to a common remote monitor. For this network, we rate the information freshness of the CSMA/CA based network using the age of information (AoI). Different from previous work, the network we consider is unsaturated. To theoretically analyze the transmission behavior of the CSMA/CA based network, we, therefore, develop an equivalent and tractable Markov transmission model. Based on this newly developed model, the transmission probability, collision probability and average AoI of the network are obtained. Our numerical results show that as the packet rate and the number of nodes increase, both the transmission probability and collision probability are increasing; the average AoI first decreases and then increases as the packet rate increases and increases with the number of nodes.
\end{abstract}


\IEEEpeerreviewmaketitle

\section{Introduction}
With the rapidly developing of Internet-of-Things (IoT) technology and the continuous advancement of Industry 4.0 systems, there are strict timeliness requirements for the increasing number of real-applications.
Traditional metrics like throughput and latency are no longer applicable, and thus a novel timeliness metric named the age of information (AoI) was first put forward in \cite{S.KaulproposeAOI2011}. Since then, the AoI, which is described as the difference between the current epoch and the generation epoch of the newest reception packet \cite{S.Kaul.definition.MM12012}, has been extensively researched. First, there are some works that have concentrated on studying the average AoI in several queueing models like M/M/1, M/D/1, and D/M/1 \cite{S.Kaul.definition.MM12012}. Moreover, AoI under various strategies like the LCFS strategy, the LGFS strategy, and so on \cite{S.K.Kaul.LCFS2012}\cite{A.M.Bedewy.LGFS.2016} were also investigated.

Due to the constrained spectrum resources of wireless networks, congestions are caused when massive wireless IoT devices access a common channel simultaneously. This would result in a remote monitor not receiving valuable update information in time and wasting limited channel resources. Thus, it is crucial that massive wireless IoT devices using suitable access protocols over a shared channel, in a timely and efficient manner, send their updating information \cite{qwang,swan} . In addition, in most practical networks, using fixed multiple access strategies, such as TDMA and FDMA, cannot meet the strict timeliness requirements of real-time applications since IoT devices may work intermittently. Therefore, random multiple access strategies, e.g., CSMA/CA and so on, are more practical for IoT networks and Industry 4.0 systems \cite{xchen,zyao}.

Traditional performance metric like throughput has been investigated in the networks employing CSMA/CA, e.g., \cite{G.Bianchi2000} and \cite{K.Duffy.non-saturated2005}. Specifically, G.Bianchi established a Markov chain transmission model for the CSMA/CA based saturated network (i.e., all nodes, in the network, always have packets.) to facilitate the analysis of network throughput \cite{G.Bianchi2000}.
Afterwards, Duffy et al. extended the Markov chain proposed by G. Bianchi to the unsaturated network and the results showed that the newly developed model can be extensively used in traffic areas \cite{K.Duffy.non-saturated2005}.

The novel performance metric AoI of the networks based on the CSMA/CA was also studied. For instance, the average AoI was first studied in \cite{S.KaulproposeAOI2011} for the CSMA/CA based a vehicular network through simulations. Motivated by this work, the authors analyzed the AoI of the network employing the CSMA/CA protocol in the case of no collisions occurring, using the stochastic hybrid systems method \cite{A.Maatouk2019}. The average AoI of a network based on CSMA/CA with large collision probabilities was also investigated in \cite{M.Moltafet2020}.

\begin{figure}[!t] %h表示这儿；t表示顶部；b表示底部；p表示本页
	\centering %居中
	\includegraphics[width=3.5in]{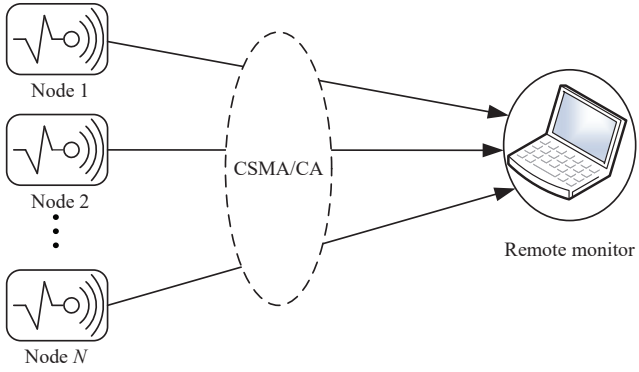}%图片大小及加载的图片名称		
	\caption{The CSMA/CA based  wireless network.}%图片标题
	\label{system_model}%标注该图片，用于在文章内引用
\end{figure}

In this paper, we shall rate the timeliness of an unsaturated wireless network via the AoI. To be specific, all nodes share a common channel by the CSMA/CA protocol to deliver their observed updates to a remote monitor and do not always have packets arriving, as shown in Fig. \ref{system_model}.

With respect to this wireless network employing the CSMA/CA, we shall build a tractable Markov transmission model to analyze the transfer process. Moreover, the average AoI of the network based on CSMA/CA is also explicitly presented in a closed form. Our numerical results show that both the transmission probability and collision probability increase with the population of nodes and packet rate; as packet rate increases, the average AoI is decreasing first and then increasing; the average AoI is increasing with the population of nodes.

We summarize the remaining organization of this paper as follows. In Section \ref{organize-system-model}, the CSMA/CA based wireless model is presented. As shown in Section \ref{organize-system-analysis}, we develop a universal Markov transmission model for the network using CSMA/CA to analyze the transmission behavior and derive the service rate of the network. We also explicitly give the average AoI in a closed form. Finally, we show numerical results and conclusions in Section \ref{organize-simulation-results} and \ref{organzie-conclusion}, respectively.

\section{System Model} \label{organize-system-model}
As shown in Fig. \ref{system_model}, we focus on a wireless network composed of $N$ independent source nodes of identical characteristics and a remote public monitor, where every node generates packets following a Bernoulli process at rate $p$. In particular, all the nodes compete for a shared channel over the CSMA/CA protocol. It is assumed that time is slotted and only one slot is required for the transmission of each packet.

In the CSMA/CA protocol, every node must wait for a back-off time $t$ slots before trying to send a packet. The back-off counter records its corresponding back-off time $t$. Note that $t$ is chosen from zero to ${{w}_{s}}$, in which we denote ${{w}_{s}}={{2}^{s}}{{w}_{0}}$ as the competition window, $s$ is the number of unsuccessful back-off stage and can be infinite (i.e., $s=0,1,\cdots$). $t$ would be reduced by one in each slot in case of no collision occurring. Otherwise, it keeps unchanged. A node will try to send its head-of-line packet if the back-off counter is zero. If a transmission is successful, the node starts the back-off process of a new packet by resetting its back-off stage $s$ to zero and a random back-off counter. If a collision occurs, the node enters the next back-off stage by doubling the contention window and a random back-off counters.

We evaluate the timeliness of the CAMA/CA based network via the AoI. The AoI is the difference between the current slot $m$ and the generation slot ${{U}_{n}}\left( m \right)$ of the newest successful reception packet. Thus, the AoI is mathematically represented as
\begin{equation}
   {{\Delta }_{n}}\left( m \right)=m-{{U}_{n}}\left( m \right).
\end{equation}
\begin{figure}[!htp]
	\centering %居中
	\includegraphics[width=3.5in]{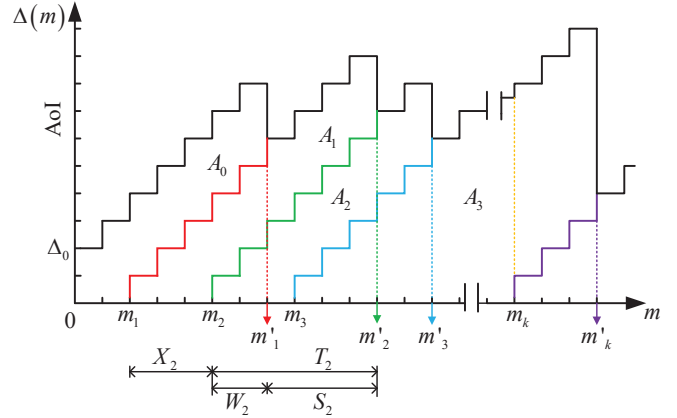}%图片大小及加载的图片名称		
	\caption{A sample paths of AoI.}%图片标题
	\label{age-sample}%标注该图片，用于在文章内引用
\end{figure}

Fig. \ref{age-sample} describes a sample path of the AoI. We denote the arrival of packet $k$ as $m_{k}$, denote the reception of the packet as $m_{k}'$, denote the arrival time of two neighboring packets as inter-arrival time $X_{k}$, denote the waiting time and service time of packet $k$, respectively, as $W_{k}$ and $S_{k}$. From Fig. \ref{age-sample}, we can represent the system time of packet $k$ as $T_{k}=W_{k}+S_{k}$.

\section{Transmission Behavior and Average AoI of CSMA/CA Based Network} \label{organize-system-analysis}
In this section, we first build an equivalent and universal Markov transmission model for this network. Second, according to this newly developed Markov model, the transmission probability, collision probability, and service rate are obtained. After that, we explicitly give the average AoI of the network in a closed form.
\subsection{Equivalent and Universal Transmission Model}
In this subsection, to facilitate the analysis of the transmission behavior of the CSMA/CA based network, based on the collision probability ${{p}_\text {cl}}$ and transmission probability ${{p}_\text {tx}}$, we develop an equivalent and universal Markov transmission model. In particular, ${{p}_\text {cl}}$ is probability that two or more nodes deliver their packets in the same slot; ${{p}_\text {tx}}$ is probability that a node attempts to transmit a packet when its corresponding back-off counter is zero. Unlike \cite{G.Bianchi2000}, in this paper, we shall study an unsaturated network, i.e., the buffer of all nodes is not always non-empty. To this end, we model a three-dimensional process $(s, t, c)$ as a Markov chain with the number $s$ of unsuccessful back-offs, the back-off counter $t$ of a node, and the number $c$ of packets in the cache as the state of the node. In particular, it is said that a node is in the idle state $(-1,-1,0)$ if its buffer is empty, i.e., $c=0$.

\begin{figure}
\centering
\subfigure[Transition diagram for each node buffer size $c = k \geq 1$.]{
\includegraphics[width=3.5in]{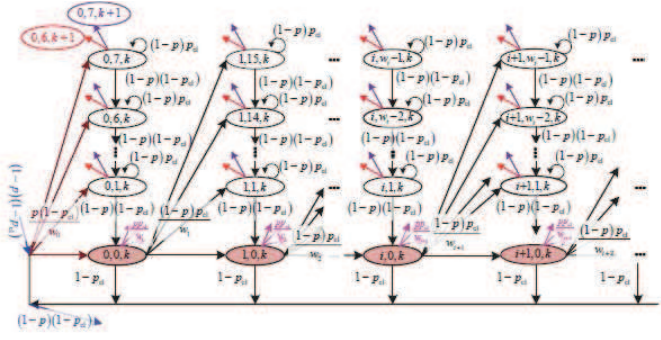}\label{Markov-a}
}
\subfigure[Transition diagram of the three-dimensional Markov chian.]{
\includegraphics[width=3.5in]{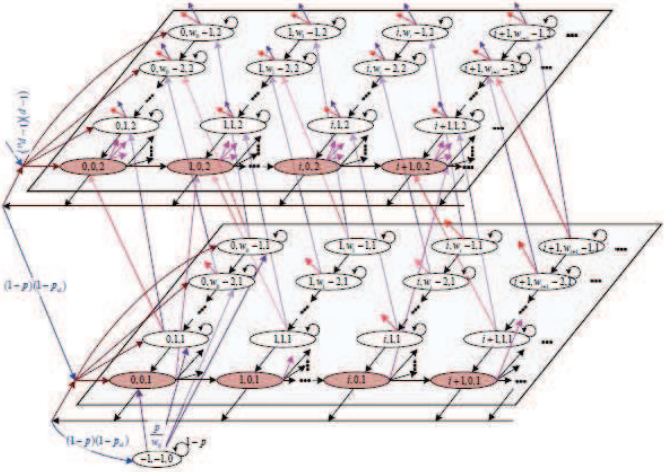}\label{Markov-b}
}
\DeclareGraphicsExtensions.
\caption{Three dimensional transition diagram for a node state $(s, t, c)$.}
\label{Markov}
\end{figure}

As shown in Fig. \ref{Markov}, we present the state transition diagram of the three-dimensional Markov chain, where Fig. \ref{Markov-a} describes the state transition diagram of the Markov chain when the buffer is non-empty (i.e., $c=k\ge 1$), and its corresponding transition probabilities are denoted as
\begin{align}\label{eq:non-empty transition probabilities}
   \left\{ \begin{aligned}
   & \Pr \left\{ i,j,k\left| i,j,\left. k \right\} \right. \right.=\left( 1-p \right){{p}_\text {cl}} \\
   & \Pr \left\{ i,j-1,k\left| i,j,\left. k \right\} \right. \right.=\left( 1-p \right)\left( 1-{{p}_\text {cl}} \right) \\
   & \Pr \left\{ 0,j,k\left| i,0,\left. k \right\} \right. \right.=\frac{p\left( 1-{{p}_\text {cl}} \right)}{{{w}_{0}}} \\
   & \Pr \left\{ i+1,j,k\left| i,0,\left. k \right\} \right. \right.=\frac{\left( 1-p \right){{p}_\text {cl}}}{{{w}_{i+1}}},
\end{aligned} \right.
\end{align}
in which $p$ is the packet rate of each node, $j\in (1,{{w}_{i}}-1)$, $i=0,1,\cdots$.

In \eqref{eq:non-empty transition probabilities}, the first and second equations describe the cases that when no packet is generated, the back-off counter will be frozen if a collision is sensed and will be reduced by one if no collision occurs. The remaining equations account for the case where the back-off counter returns to zero. If a node delivers its packet without collisions (a packet is removed) and a new packet arrives in the slot (this new packet is pushed into the buffer), the node will start a new back-off process in the same layer. If a collision is detected and no packet arrives, the node will enter the next back-off stage with a random back-off counter.

Fig. \ref{Markov-b} shows the state transition diagram when the node buffer changes and its corresponding transition probabilities are given by
\begin{align}\label{eq:transition probabilities}
   \left\{
   \begin{aligned}
   & \Pr \left\{ -1,-1,0\left| -1,-1,\left. 0 \right\} \right. \right.=1-p \\
   & \Pr \left\{ 0,j,1\left| -1,-1,\left. 0 \right\} \right. \right.=\frac{p}{{{w}_{0}}} \\
   & \Pr \left\{ -1,-1,0\left| i,0,\left. 1 \right\} \right. \right.=(1-p)(1-{{p}_\text {cl}}) \\
   & \Pr \left\{ 0,j,k\left| i,\left. 0,k \right\} \right. \right.=\frac{(1-p)(1-{{p}_\text {cl}})}{{{w}_{0}}} & k\ge 2& \\
   & \Pr \left\{ i+1,j,k+1\left| i,\left. 0,k \right\} \right. \right.=\frac{p{{p}_\text {cl}}}{{{w}_{i+1}}}  &k\ge 1& \\
   & \Pr \left\{ i,j,k+1\left| i,j,\left. k \right\} \right. \right.=p{{p}_\text {cl}} & k\ge 1& \\
   & \Pr \left\{ i,j-1,k+1\left| i,j,\left. k \right\} \right. \right.=p(1-{{p}_\text {cl}}) & k\ge 1&.
  \end{aligned}
  \right.
\end{align}

In \eqref{eq:transition probabilities}, the first equation describes the fact that a node is in the idle state. The second equation explains the situation that a node generates a packet and will transit to the zeroth back-off stage with a random back-off counter.
The third equation describes the case when a node successfully delivers its only $(k=1)$ packet and no new packet arrives. Likewise, the fourth equation accounts for the situation that a node successfully delivers its head-of-line packet $(k\ge 2)$ and no new packet arrives. In this case, the back-off process of the followed packet restarts from the next layer. The fifth equation characterizes the case when the back-off counter returns to zero and a collision is sensed as well as a new packet arrives. Thus, the state will transit to the next stage in the above layer. Finally, the back-off counter is non-zero and the node gets a new packet, as shown in the sixth and seventh equations.

We denote the stationary distribution of the Markov chain as ${{b}_{i,j,k}}=\underset{m\to \infty }{\mathop{\lim }}\,\Pr \left\{ (s,t,c)=(i,j,k) \right\}$ , ${{b}_{i,j,*}}=\sum\nolimits_{k=1}^{\infty }{{{b}_{i,j,k}}}$, and the service rate of the network using the CSMA/CA protocol as $\mu$. In particular, the stationary distribution $\boldsymbol{b}$ is given by the below proposition.
\begin{proposition}\label{proposition-the-stationary-distribution-of-the-Markovchain}
Given a packet rate $p$ and collision probability ${{p}_\text {cl}}$, we can get the $\boldsymbol{b}$
\begin{align}
{{b}_{0,0,*}}&=p, \\
{{b}_{i,0,*}}&=pp_\text {cl}^{i}, \\
{{b}_{i,j,*}}&=\frac{p\left( {{w}_{i}}-i \right)p_\text {cl}^{i}}{{{w}_{i}}\left( 1-{{p}_\text {cl}} \right)}, \\
{{b}_\text{idle}}&=1-\frac{p\left( 4p_\text {cl}^{2}-\left( {{w}_{0}}+4 \right){{p}_\text {cl}}+{{w}_{0}}+1 \right)}{2{{\left( 1-{{p}_\text {cl}} \right)}^{2}}\left( 1-2{{p}_\text {cl}} \right)},
\end{align}
in which $i\ge 1$, $0\le j\le {{w}_{i}}-1$, and ${{w}_{i}}={{2}^{i}}{{w}_{0}}$. Moreover, the idle probability is given by
\begin{align}\label{eq:csma-bidle}
{{p}_\text{idle}}=1-\frac{p\left( 4p_\text {cl}^{2}-\left( {{w}_{0}}+4 \right){{p}_\text {cl}}+{{w}_{0}}+1 \right)}{2{{\left( 1-{{p}_\text {cl}} \right)}^{2}}\left( 1-2{{p}_\text {cl}} \right)}.
\end{align}
\end{proposition}

\begin{proof}
See Appendix \ref{proof-stationary-distribution-of-the-markovchain}.
\end{proof}

According to Proposition \ref{proposition-the-stationary-distribution-of-the-Markovchain}, the transmission probability ${{p}_\text{tx}}$ and collision probability ${{p}_\text {cl}}$ are presented by the below theorem.
\begin{theorem}\label{theorem-transmission-and-collision-probability}
Given the packet rate $p$ and stationary distribution $\boldsymbol{b}$ of the Markov chain, ${{p}_\text{tx}}$ and ${{p}_\text {cl}}$ are given by
\begin{align}\label{eq:ptx-probability}
{{p}_\text {tx}}&=\frac{p}{1-{{p}_\text {cl}}}
\end{align}
\begin{align}\label{eq:pcl-probability}
{{p}_\text {cl}}&=1-{{(1-{{p}_\text {tx}})}^{N-1}},
\end{align}
in which $N$ is the number of nodes.
\end{theorem}
\begin{proof}

See Appendix \ref{proof-transmission-and-collision-probability}.
\end{proof}
By combining \eqref{eq:ptx-probability} and \eqref{eq:pcl-probability}, ${{p}_\text {tx}}$ can be obtained with the Newton's iteration method. After that, ${{p}_\text{cl}}$ (cf.\eqref{eq:pcl-probability}) and ${{p}_\text{idle}}$ (cf.\eqref{eq:csma-bidle}) can also be obtained.

Note that the probability for a node to deliver a packet successfully in each slot is expressed as
\begin{align}
   {{p}_\text{s}}={{p}_\text{tx}}\left( 1-{{p}_\text {cl}} \right)=p.
 \end{align}
Therefore, the service rate can be acquired
\begin{align}\label{eq:service-rate-mu}
   {\mu } {=}\frac{p}{1-{{p}_\text{idle}}}.
\end{align}

We represent the time that a node delivers a packet successfully as the service time ${{S}_{k}}$ and model ${{S}_{k}}$ as a geometrically distributed random variable with a mean of $1/\mu$. Thus, we have
\begin{align}\label{eq:service-time-distributed}
\Pr \left\{ {{S}_{k}}=j \right\}={\mu } {{\left( 1-{\mu }  \right)}^{j-1}}, j=1,2,\cdots.
\end{align}
The expression \eqref{eq:service-time-distributed} can be explained as follows. In each slot of the transmission, the service would be completed with probability $\mu$ at the end of that slot and would be continued in the next slot with probability $1-\mu$.
\begin{remark}
\begin{itemize}
To obtain the maximum packet rate $p_\text{max}$ and number $N_\text{max}$ of nodes of network, we set $p\to \mu$, thus we approximately have ${{p}_{\text{idle}}}=0$ from \eqref{eq:service-rate-mu}.
  \item For a fixed $N$, according to $p_\text{idle}$ (cf. \eqref{eq:csma-bidle}) and by substituting \eqref{eq:pcl-probability} into \eqref{eq:csma-bidle}, we have
  \begin{align}\label{eq:ptx-equation}
     \frac{{{p}_{\text{tx}}}\left( 4p_{\text{cl}}^{\text{2}}-\left( {{w}_{0}}+4 \right){{p}_{\text{cl}}}+{{w}_{0}}+1 \right)}{2{{(1-{{p}_{\text{tx}}})}^{N-1}}\left( 1-2{{p}_{\text{cl}}} \right)}=1,
  \end{align}
  where ${{p}_{\text{cl}}}=1-{{\left( 1-{{p}_{\text{tx}}} \right)}^{N-1}}$. Specifically, we can get the real solution $\overline{{{p}_{\text{tx}}}}$ of $p_\text{tx}$ from \eqref{eq:ptx-equation}. By combining \eqref{eq:pcl-probability} to \eqref{eq:ptx-probability}, and replacing ${p}_{\text{tx}}$ with $\overline{{{p}_{\text{tx}}}}$, we have
  \begin{align}\label{eq:p-max}
     {p}_{\text{max}}=\overline{{{p}_{\text{tx}}}}{{\left( 1-\overline{{{p}_{\text{tx}}}} \right)}^{N-1}}.
  \end{align}
  \item Likewise, given $p$ and based on ${{p}_{\text{idle}}}=0$, we can obtain the real solution $\overline{{{p}_{\text{cl}}}}$ of $p_\text{cl}$ (cf. \eqref{eq:csma-bidle}). Therefore, $N_\text{max}$ can be expressed as
  \begin{align}
  {N}_\text{max}=\lfloor \frac{\ln \left( 1-\overline{{{p}_{\text{cl}}}} \right)}{\ln \left( 1-{{p}_{\text{tx}}} \right)}+1 \rfloor,
  \end{align}
  where $\lfloor x \rfloor$ denotes the largest integer not exceeding $x$.
\end{itemize}
\end{remark}
\subsection{Average AoI of the Network}
For each node, the packets generated are a Bernoulli process with parameter $p$ and service time ${S}_{k}$ is the geometric distribution (cf.\eqref{eq:service-time-distributed}), therefore, the queueing process at each node would be a Geom/Geom/1 queue. Under this model, the probability generating function (PGF) of the system time of a packet is presented by the below lemma.

\begin{lemma}\label{PGF-of-system-time-proposition}
Given $p<\mu$, the PGF of the system time is given by
\begin{align}\label{eq:PGF-of-the-system-time-expression}
   {{G}_{T}}\left( z \right)=\frac{\beta z}{1-\left( 1-\beta  \right)z}, \quad \beta =\frac{\mu -p}{1-p}.
\end{align}
\end{lemma}

\begin{proof}
A detailed proof can be found in \cite[Chap. 3.1.2]{N.TianandX.Xu2008}.
\end{proof}

Further, the average AoI is presented by the below theorem.
\begin{theorem}\label{AOI-expression-theorem}
Given the packet rate $p$, the average AoI of the network can be represented as
\begin{align}
   & \bar{\Delta}=\frac{1}{p}+\frac{p}{\mu }+\frac{1-p}{\mu -p}-\frac{p}{{{\mu }^{2}}}.
\end{align}
\end{theorem}
\begin{proof}
See Appendix \ref{proof-AOI-expression-theorem}.
\end{proof}

By Theorem \ref{AOI-expression-theorem}, the average AoI of the CSMA/CA based network can be acquired, where the service rate $\mu$ is given by \eqref{eq:service-rate-mu}.

\section{Numerical Results} \label{organize-simulation-results}
In this section, we research the average AoI through the numerical results. We set $w_{0}=8$.
\begin{figure}[!htp] %h表示这儿；t表示顶部；b表示底部；p表示本页
	\centering %居中
	\includegraphics[width=3.5in]{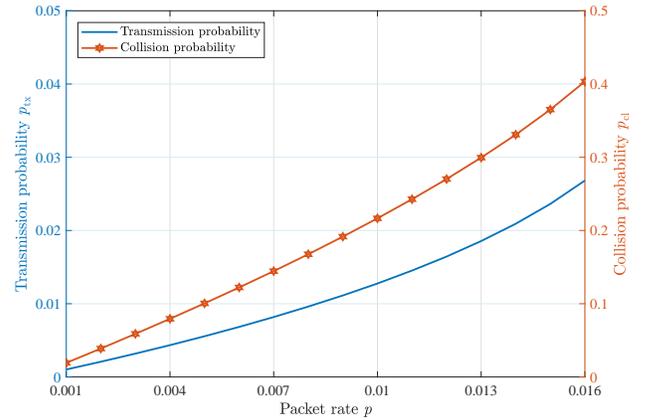}%图片大小及加载的图片名称		
	\caption{$p_\text{tx}$ and $p_\text{cl}$ versus $p$ ($N=20$)}%图片标题
	\label{p-change}%标注该图片，用于在文章内引用
\end{figure}

\begin{figure}[!htp] %h表示这儿；t表示顶部；b表示底部；p表示本页
	\centering %居中
	\includegraphics[width=3.5in]{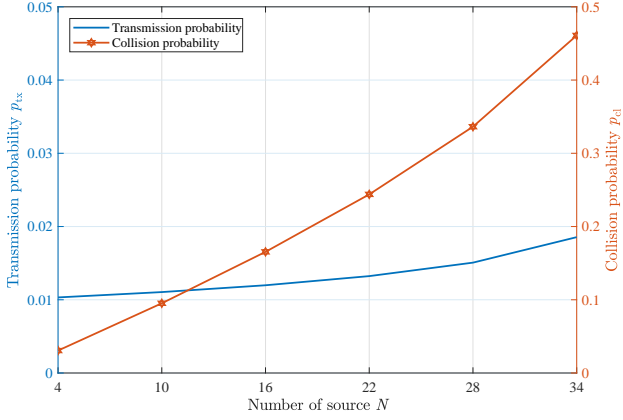}%图片大小及加载的图片名称		
	\caption{$p_\text{tx}$ and $p_\text{cl}$ versus $N$ ($p=0.01$)}%图片标题
	\label{node-change}%标注该图片，用于在文章内引用
\end{figure}
As shown in Figs. \ref {p-change} and \ref{node-change}, we present both the transmission probability $p_\text{tx}$ and collision probability $p_\text{cl}$ as functions of packet rate $p$ and the number $N$ of nodes. First, we observe that $p_\text{tx}$ increases with both $p$ and $N$. This is because $p$ increasing leads to a larger probability of non-empty buffer, which further leads to a larger $p_\text{tx}$. Likewise, an increase in $N$ results in more attempts of transmission and thus a larger $p_\text{tx}$. Moreover, we see also that $p_\text{cl}$ increases with both $p$ and $N$. This is because as $p$ and $N$ increase, there will be more competitions in the transmissions.
\begin{figure}[!htp] %h表示这儿；t表示顶部；b表示底部；p表示本页
	\centering %居中
	\includegraphics[width=3.5in]{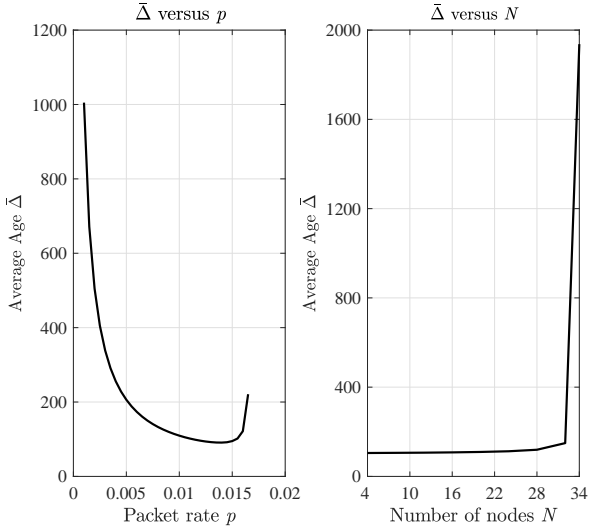}%图片大小及加载的图片名称		
	\caption{Change in the average AoI}%图片标题
	\label{AoI-change}%标注该图片，用于在文章内引用
\end{figure}

In Fig. \ref{AoI-change}, we show how the packet rate $p$ and number $N$ of nodes affect the average AoI. First, we see that the average AoI is decreasing first and then increasing with $p$ increasing. In particular, when $p$ is relatively small or large, the average AoIs are quite large. The main reason for this result is that the new updates waiting times are longer when $p$ is small, as well as when $p$ is large the frequency of collisions becomes higher, thus the service is longer. Second, we also see that as $N$ increases, the average AoI is increasing. This is a consequence of the collision frequency and service time becoming larger with $N$ growth.
\section{Conclusion} \label{organzie-conclusion}
In this paper, the timeliness of wireless network based on CSMA/CA has been investigated. Particularly, we are interested in how the average AoIs vary with the packet rate and the number of nodes. We showed that the average AoI would be larger when a packet rate is relatively small or large and average AoI increases with the number of nodes. With respect to this study, we considered that it takes exactly one slot for the transmission of each packet. The case where more transmission slots are required each packet would be investigated in future work. Moreover, we shall further study the mixed use of fixed and random multiple access technologies.

\appendix
\subsection{Proof of Proposition \ref{proposition-the-stationary-distribution-of-the-Markovchain}}\label{proof-stationary-distribution-of-the-markovchain}
\begin{proof}
First, given $p<\mu$, in the zeroth back-off stage, we know that the back-off counter will definitely reaches zero eventually and stay at state $(0, 0, k)$ for some $k\ge 1$ for exactly one slot. Thus, we have ${{b}_{0,0,*}}=p$.

Second, if a collision occur, a state will transfer from the $(i,0,*)$ to the $(i+1)$-th back-off stage. At this time, the state returns and remains at $(i+1,0)$ for exactly one slot, and thus we have
\begin{align}\label{eq:back-off-stage-from-i-i+1}
   {{b}_{i,0,*}}{{p}_\text {cl}}={{b}_{i+1,0,*}}, \quad i\ge 0.
\end{align}
Further, we have
\begin{align}\label{eq:stationary-distribution-b-i,0,*}
   {{b}_{i,0,*}}=p_\text {cl}^{i}{{b}_{0,0,*}}=pp_\text {cl}^{i}, \quad i\ge 0.
\end{align}
It is easy to prove \eqref{eq:stationary-distribution-b-i,0,*} by the dynamic equilibrium property of stationary Markov chains. To this end, we shall divide the state space into the states $(i',j,k)$ with $i'\leq i$ and the states $(i'',j,k)$ with $i'' \ge i$. Since the total transition probabilities of the left-to-right part and the right-to-left part are equal, thus we have
\begin{align}\label{eq:the-sum-of-the-transition-probabilities-of-the-right-part}
   {{b}_{i,0,*}}{{p}_\text {cl}}&=\sum\limits_{i''=i+1}^{\infty }{{{b}_{i'',0,*}}(1-{{p}_\text {cl}})}\nonumber \\
                                &={{b}_{i+1,0,*}}(1-{{p}_\text {cl}})+\sum\limits_{i''=i+2}^{\infty }{{{b}_{i'',0,*}}(1-{{p}_\text {cl}})},
\end{align}
which yields
\begin{align}\label{eq:the-sum-of-the-transition-probabilities-of-the-right-part-from-i+2}
   {{b}_{i+1,0,*}}{{p}_\text {cl}}=\sum\limits_{i''=i+2}^{\infty }{{{b}_{i'',0,*}}(1-{{p}_\text {cl}})}.
\end{align}
By substituting \eqref{eq:the-sum-of-the-transition-probabilities-of-the-right-part-from-i+2} into \eqref{eq:the-sum-of-the-transition-probabilities-of-the-right-part}, \eqref{eq:back-off-stage-from-i-i+1} can be obtained.
We denote the initial back-off counter of the i-th back-off stage as $J$ and have
\begin{align}
   \Pr \left\{ J\ge j \right\}=\frac{{{w}_{i}}-j}{{{w}_{i}}}, 0\le j\le {{w}_{i}}-1.
\end{align}
Note that a node has a chance to stay in state $(i,j,*)$ as long as the initial back-off counter $J$ is no less than $j$. When the back-off counter has been reduced to $j$, the expected sojourn time for this state could be denoted as
\begin{align}
   \bar{L}=\sum\nolimits_{l=1}^{\infty }{l\left( 1-{{p}_{\text{cl}}} \right)}p_{\text{cl}}^{l-1}=\frac{1}{1-{{p}_{\text{cl}}}}.
\end{align}
Afterwards, the state eventually transits to $(i,0,*)$ and stay for only one slot. Thus, the stationary probability that a node is in state $(i,j,*)$ is represented by
\begin{align}\label{back-off-stationary-distributed}
   {{b}_{i,j,*}}=\bar{L}\Pr \left\{ J\ge j \right\}{{b}_{i,0,*}}=\frac{p\left( {{w}_{i}}-j \right)p_\text {cl}^{i}}{{{w}_{i}}\left( 1-{{p}_\text {cl}} \right)},
\end{align}
in which $i\ge 0$, $0\le j\le {{w}_{i}}-1$, and \eqref{eq:stationary-distribution-b-i,0,*} is used. Moreover, \eqref{back-off-stationary-distributed} can also be proved by applying the dynamic equilibrium equation to the Markov chain.

Finally, by the normalized condition of the stationary distribution, we have
\begin{align}
   {{b}_\text{idle}}
   &={{b}_{-1,-1,0}}=1-\sum\limits_{i=0}^{\infty }{{{b}_{i,0,*}}-\sum\limits_{i=0}^{\infty }{\sum\limits_{j=1}^{{{w}_{i}}-1}{{{b}_{i,j,*}}}}}\nonumber \\
   &=1-\frac{p\left( 4p_\text {cl}^{2}-\left( {{w}_{0}}+4 \right){{p}_\text {cl}}+{{w}_{0}}+1 \right)}{2{{\left( 1-{{p}_\text {cl}} \right)}^{2}}\left( 1-2{{p}_\text {cl}} \right)}.
\end{align}

This finishes the proof of Proposition \ref{proposition-the-stationary-distribution-of-the-Markovchain}.
\end{proof}
\subsection{Proof of Theorem \ref{theorem-transmission-and-collision-probability}}\label{proof-transmission-and-collision-probability}
\begin{proof}
Based on Proposition \ref{proposition-the-stationary-distribution-of-the-Markovchain} and the fact that a node attempts to deliver a packet every time when its back-off counter is zero. Thus, the transmission probability is the total probability that a node is in a state $(i,0,k)$, for $i\geq0$,$k\geq1$ and we have
\begin{align}\label{eq:proof-csma-ptx}
    {{p}_\text{tx}}=\sum\limits_{i=0}^{\infty }{{{b}_{i,0,*}}}=\sum\limits_{i=0}^{\infty }{p_\text {cl}^{i}}{{b}_{0,0,*}}=\frac{p}{1-{{p}_\text {cl}}}.
\end{align}
The collision probability is the conditional probability that a reference node transmits it packets while one or more of the remaining $N-1$ nodes send their packets at the same time. Thus, ${p}_\text {cl}$ is denoted as
\begin{align}
   {{p}_\text {cl}}=1-{{\left( 1-{{p}_\text {tx}} \right)}^{N-1}}.
\end{align}

This finishes the proof of Theorem \ref{theorem-transmission-and-collision-probability}.
\end{proof}

\subsection{Proof of Theorem \ref{AOI-expression-theorem}}\label{proof-AOI-expression-theorem}
\begin{proof}
Consider a period of $M$ slots, in which the monitor receive successfully $K$ packets. By dividing the area of Fig. \ref{age-sample} into a series of triangular-like $A_{0}, A_{1}, \cdots$, thus the average AoI is denoted as
\begin{align}\label{eq:average-AoI-expression}
   \bar{\Delta} &=\underset{M\to \infty }{\mathop{\lim }}\,\frac{k-1}{M}\frac{1}{k-1}\sum\limits_{k=1}^{K}{{{A}_{k}}} \nonumber \\
                &=p\mathbb{E}\left[ {{A}_{k}} \right]
\end{align}
where $p=\underset{M\to \infty }{\mathop{\lim }}\,{K}/{M}\;$ is the packet rate. Moreover,
\begin{align}\label{eq:the-average-area}
   \mathbb{E}\left[ {{A}_{k}} \right]
   & =\frac{1}{2}\mathbb{E}\left[ X_{k}^{2} \right]+\frac{1}{2}\mathbb{E}\left[ {{X}_{k}} \right]+\mathbb{E}\left[ {{X}_{k}}{{T}_{k}} \right].
\end{align}

Since the packets generated are a Bernoulli process with rate $p$ and inter-arrival time ${X}_{k}$ would be a geometrically distributed random variable with a mean of $1/p$. Thus, we can get
\begin{align}\label{eq:the-first-moment}
   \mathbb{E}\left[ {{X}_{k}} \right]=\frac{1}{p}, \quad \mathbb{E}\left[ X_{k}^{2} \right]=\frac{2-p}{{{p}^{2}}}.
\end{align}

To derive $\mathbb{E}\left[ {{X}_{k}}{{T}_{k}} \right]= \mathbb{E}\left[ {{X}_{k}}{{W}_{k}} \right]+\mathbb{E}\left[ {{X}_{k}} \right]\mathbb{E}\left[ {{S}_{k}} \right]$, we introduce the auxiliary function $H\left( z \right)$ as follows(cf.\cite{Y.Dong2019})
\begin{align}\label{eq:define-auxiliary-function }
   H\left( z \right)
   & =\sum\limits_{j=1}^{\infty }{{{z}^{j}}}\Pr \left\{ {{X}_{k}}=j \right\}\sum\limits_{i=0}^{\infty }{\Pr \left\{ {{T}_{k-1}}>j+i \right\}} \nonumber \\
   & =\frac{\left( 1-\beta  \right)pz}{\beta \left( 1-\left( 1-\beta  \right)\left( 1-p \right)z \right)},
\end{align}
where $\Pr \left\{ {{T}_{k-1}}>i \right\}={{\left( 1-\beta  \right)}^{i}}$.
Moreover,
\begin{align}\label{eq:calculate-X_k-W_k}
    \mathbb{E}\left[ {{X}_{k}}{{W}_{k}} \right]
       & =\sum\limits_{j=1}^{\infty }{j\Pr \left\{ {{X}_{k}}=j \right\}}\sum\limits_{i=0}^{\infty }{\Pr \left\{ \max \left( 0,{{T}_{k-1}}-j \right)>i \right\}} \nonumber \\
    & =\sum\limits_{j=1}^{\infty }{j\Pr \left\{ {{X}_{k}}=j \right\}}\sum\limits_{i=0}^{\infty }{\Pr \left\{ {{T}_{k-1}}>j+i \right\}} \nonumber \\
    & =\underset{z\to {{1}^{-}}}{\mathop{\lim }}\,{{\left( H\left( z \right) \right)}^{'}} 
     =\frac{p\left( 1-\beta  \right)}{\beta {{\left( 1-\left( 1-p \right)\left( 1-\beta  \right) \right)}^{2}}}.
\end{align}

According to \eqref{eq:average-AoI-expression} to \eqref{eq:calculate-X_k-W_k}, the average AoI given in (18) can be obtained. 
Thus, the proof of Theorem \ref{AOI-expression-theorem} has been finished.
\end{proof}

\ifCLASSOPTIONcaptionsoff
  \newpage
\fi

\end{document}